\pdfoutput=1
\documentclass[11pt,letterpaper]{article} 

\usepackage[top=80pt,bottom=80pt,left=86pt,right=84pt]{geometry}

\usepackage{graphicx}
\usepackage{microtype}
\usepackage{cite}  
\usepackage{enumerate}
\usepackage{wrapfig}
\usepackage{amsthm}
\usepackage{amsmath}
\usepackage{amssymb}
\usepackage{thm-restate}

\usepackage{color}

\usepackage{authblk}
\usepackage[colorlinks,linkcolor=blue,filecolor=blue,citecolor=blue,urlcolor=blue,pdfstartview=FitH]{hyperref}

\usepackage[nameinlink]{cleveref}

\graphicspath{{./figures/}} 

\theoremstyle{plain}
\newtheorem{theorem}{Theorem}
\newtheorem{lemma}[theorem]{Lemma}

\newtheorem{observation}[theorem]{Observation}

\newtheorem{problem}[theorem]{Problem}

\def\A{\mathcal{A}}
\def\C{\mathcal{C}}
\def\D{\mathcal{D}}
\def\G{\mathcal{G}}
\def\P{\mathcal{P}}
\def\S{\mathcal{S}}
\def\T{\mathcal{T}}

\newcommand{\frechet}{Fr\'echet}
\newcommand{\dfd}{d_{dF}}

\newcommand{\eps}{\varepsilon}
\newcommand{\reals}[1]{\mathbb{R}}
\DeclareMathOperator{\polylog}{\mathrm{polylog}}
\DeclareMathOperator{\width}{\mathit{\delta_x}}
\DeclareMathOperator{\height}{\mathit{\delta_y}}

\newcommand{\nn}{\texorpdfstring{{\sc NNC}}{NNC}}
\newcommand{\ann}{\texorpdfstring{{\sc ANNC}}{ANN}}
\newcommand{\klc}{\texorpdfstring{{\sc $(k,\ell)$-Center}}{(k,l)-Center}}
\newcommand{\otc}{\texorpdfstring{{\sc $(1,2)$-Center}}{(1,2)-Center}}

\date{}

\begin{document}
\title{Efficient Nearest-Neighbor Query and Clustering of\\ Planar Curves}

\author[1]{Boris Aronov\thanks{B.~Aronov was supported by NSF grants CCF-12-18791 and CCF-15-40656, and by grant 2014/170 from the US-Israel Binational Science Foundation.}}
\author[2]{Omrit Filtser\thanks{O.~Filtser was supported by the Israeli Ministry of Science, Technology \& Space, and by grant 2014/170 from the US-Israel Binational Science Foundation.}}
\author[3]{Michael Horton\thanks{Most of the work on this project by M.~Horton was performed while visiting the Department of Computer Science and Engineering at the Tandon School of Engineering, New York University in the spring/summer of 2018, partially supported by NSF grants CCF-12-18791.}}
\author[2]{Matthew J. Katz\thanks{M.~Katz was supported by grant 1884/16 from the Israel Science Foundation and by grant 2014/170 from the US-Israel Binational Science Foundation.
Part of the work on this project by M.~Katz was performed while visiting the  Department of Computer Science and Engineering at the Tandon School of Engineering, New York University in the spring of 2018, partially supported by NSF grants CCF-12-18791 and CCF-15-40656.
}}
\author[1]{\\Khadijeh Sheikhan}
\affil[1]{
	Tandon School of Engineering, New York University, Brooklyn, NY 11201, USA\protect\\
	Email: \texttt{boris.aronov@nyu.edu,khadije.sheikhan@gmail.com}}

\affil[2]{
Ben-Gurion University of the Negev, Beer-Sheva 84105, Israel\protect\\
Email: \texttt{\{omritna,matya\}@cs.bgu.ac.il}}

\affil[3]{
SPORTLOGiQ Inc., Montreal, Quebec H2T 3B3, Canada\protect\\
Email: \texttt{michael.horton@sportlogiq.com}}

\maketitle              %
\begin{abstract}
	
  We study two fundamental problems dealing with curves in the plane, namely, the nearest-neighbor problem and the center problem. Let $\C$ be a set of $n$ polygonal curves, each of size $m$. In the nearest-neighbor problem, the goal is to construct a compact data structure over $\C$, such that, given a query curve $Q$, one can efficiently find the curve in $\C$ closest to $Q$. In the center problem, the goal is to find a curve $Q$, such that the maximum distance between $Q$ and the curves in $\C$ is minimized. We use the well-known discrete \frechet\ distance function, both under~$L_\infty$ and under $L_2$, to measure the distance between two curves. 
	
  For the nearest-neighbor problem, despite discouraging previous results, we identify two important cases for which it is possible to obtain practical bounds, even when $m$ and $n$ are large. In these cases, either $Q$ is a line segment or $\C$ consists of line segments, and the bounds on the size of the data structure and query time are nearly linear in the size of the input and query curve, respectively. The returned answer is either exact under $L_\infty$, or approximated to within a factor of $1+\eps$ under~$L_2$. We also consider the variants in which the location of the input curves is only fixed up to translation, and obtain similar bounds, under $L_\infty$.
	
	As for the center problem, we study the case where the center is a line segment, i.e., we seek the line segment that represents the given set as well as possible. We present near-linear time exact algorithms under $L_\infty$, even when the location of the input curves is only fixed up to translation. Under $L_2$, we present a roughly $O(n^2m^3)$-time exact algorithm. 
	
\end{abstract}

\section{Introduction}
\label{sec:intro}

We consider efficient algorithms for two fundamental data-mining problems for sets of polygonal curves in the plane: nearest-neighbor query and clustering.
Both of these problems have been studied extensively and bounds on the running time and storage consumption have been obtained.
In general, these bounds suggest that the existence of algorithms that can efficiently process large datasets of curves of high complexity is unlikely.
Therefore we study special cases of the problems where some curves are assumed  to be directed line segments (henceforth referred to as segments), and the distance metric is the discrete \frechet\ distance.

Such analysis of curves has many practical applications, where the position of an object as it changes over time is recorded as a sequence of readings from a sensor to generate a \emph{trajectory}. For example, the location readings from GPS devices attached to migrating animals~\cite{ABBKKW14}, the traces of players during a football match captured by a computer vision system~\cite{Gudmundsson17}, or stock market prices~\cite{NW13}.
In each case, the output is an ordered sequence $C$ of $m$ vertices (i.e., the sensor readings), and by interpolating the location between each pair of vertices as a segment, a polygonal chain is obtained.

Given a collection $\C$ of $n$ curves, a natural question to ask is whether it is possible to preprocess $\C$ into a data structure so that the nearest curve in the collection to a query curve~$Q$ can be determined efficiently. 
This is the \emph{nearest-neighbor} problem for curves (\nn).

Indyk~\cite{Indyk02} gave a near-neighbor data structure for polygonal curves under the discrete \frechet\ distance. The data structure achieves an approximation factor of $O(\log m + \log \log n)$, where $n$ is the number of curves and $m$ is the maximum size of a curve. Its space consumption is very high, $O(|X|^{\sqrt{m}} (m^{\sqrt{m}}n)^2)$, where $|X|$ is the size of the domain on which the curves are defined, and the query time is $O(m^{O(1)} \log n)$. 

Later, Driemel and Silvestri~\cite{DS17} presented a locality-sensitive-hashing scheme for curves under the discrete \frechet\ distance, improving the result of Indyk for short curves. They also provide a trade-off between approximation quality and computational performance: for a parameter $k\in[m]$, a data structure using $O(2^{2k}m^{k-1} n\log n + mn)$ space is constructed that answers queries in $O(2^{2k}m^k \log n)$ time with an approximation factor of $O(m/k)$.

Recently, Emiris and Psarros~\cite{EmirisP18} presented near-neighbor data structures for curves under both discrete \frechet\ and dynamic time warping distance. Their algorithm achieves approximation factor of $1+\eps$, at the expense of increasing space usage and preprocessing time. For curves in the plane, the space used by their data structure is $\tilde{O}(n)\cdot(2+\frac{1}{\log m})^{O(m^{1/\eps}\cdot \log(1/\eps))}$   for discrete \frechet\ distance and $\tilde{O}(n)\cdot O(\frac{1}{\eps})^{m}$ for dynamic time warping distance, while the query time in both cases is $O(2^{2m}\log n)$. 
	
De Berg et al.~\cite{dBGM17} described a dynamic data structure for approximate nearest neighbor for curves (which can also be used for other types of queries such as approximate range searching), under the (continuous) \frechet\ distance. 
Their data structure uses $n\cdot O\left(\frac{1}{\eps}\right)^{2m}$ space and has $O(1)$ query time (for a segment query), but with an \emph{additive} error of $\eps \cdot reach(Q)$,  where $reach(Q)$ is the maximum distance between the start vertex of the query curve $Q$ and any other vertex of~$Q$. Furthermore, when the distance from $Q$ to its nearest neighbor is relatively large, the query procedure might fail.

Afshani and Driemel~\cite{AD18} studied range searching under both the discrete and continuous \frechet\ distance. In this problem, the goal is to preprocess $\C$ such that, given a query curve $Q$ of length $m_q$ and a radius $r$, all curves in $\C$ that are within distance $r$ of $Q$ can be found efficiently.
For the discrete \frechet\ distance in the plane, their data structure uses space in $O(n(\log \log n)^{m-1})$ and has query time in $O(\sqrt{n}\cdot\log^{O(m)}n\cdot m_q^{O(1)})$, assuming $m_q = \log^{O(1)}n$. They also show that any data structure in the pointer model that achieves $Q(n) + O(k)$ query time, where $k$ is the output size, has to use roughly $\Omega(n/Q(n))^2)$ space in the worst case,
even if queries are just points, for discrete \frechet\ distance!

De Berg, Cook, and Gudmundsson~\cite{dBCG13} considered range counting queries for curves under the continuous \frechet\ distance. 
Given a single polygonal curve $C$ with $m$ vertices, they show how to preprocess it into a data structure in $O(k \polylog m)$ time and space, so that, given a query segment $s$, one can return a constant approximation of the number of subcurves of $C$ that lie within distance $r$ of $s$ in $O(\frac{n}{\sqrt{k}} \polylog m)$ time, where $k$ is a parameter between $m$ and $m^2$.

Driemel and Har-Peled~\cite{DH12} preprocess a curve $C$ into a data structure of linear size, which, given a query segment $s$ and a subcurve of $C$, returns a $(1+\eps)$-approximation of the distance between $s$ and the subcurve in logarithmic~time.

Clustering is another fundamental problem in data analysis that aims to partition an input collection of curves into clusters where the curves within each cluster are similar in some sense, and a variety of formulations have been proposed~\cite{ACMM03,CL07,DKS16}.
The $k$-\textsc{Center} problem~\cite{Gonzalez85,AP02,HN79} is a classical problem in which a point set in a metric space is clustered.
The problem is defined as follows: given a set $\P$ of $n$ points, find a set $\G$ of $k$ center points, such that the maximum distance from a point in $\P$ to a nearest point in $\G$ is minimized.

Given an appropriate metric for curves, such as the discrete \frechet\ distance, one can define a metric space on the space of curves and then use a known algorithm for point clustering.
The clustering obtained by the \textsc{$k$-Center} problem is useful in that it groups similar curves together, thus uncovering a structure in the collection, and furthermore the center curves are of value as each can be viewed as a representative or exemplar of its cluster, and so the center curves are a compact summary of the collection.
However, an issue with this formulation, when applied to curves, is that the optimal center curves may be \emph{noisy}, i.e., the size of such a curve may be linear in the total number of vertices in its cluster, see~\cite{DKS16} for a detailed description. 
This can significantly reduce the utility of the centers as a method of summarizing the collection, as the centers should ideally be of low complexity.
To address this issue, Driemel et~al.~%
\cite{DKS16} introduced the \klc{} problem, where the $k$~desired center curves are limited to at most $\ell$ vertices each.

Inherent in both problems is a notion of \emph{similarity} between pairs of curves, which is expressed as a
distance function.
Several such functions have been proposed to compare curves, including the continuous~\cite{Frechet1906,AG95} and discrete~\cite{EM94} \frechet{} distance, the Hausdorff distance~\cite{Hausdorff27}, and dynamic time warping~\cite{BC94}.
We consider the problems under the discrete \frechet\ distance, which is often informally described by two frogs, each hopping from vertex to vertex along a polygonal curve.
At each step, one or both of the frogs may advance to the next vertex on its curve, and then the distance between them is measured using some point metric.
The discrete \frechet\ distance is defined as the smallest maximum distance between the frogs that can be achieved in such a joint sequence of hops of the frogs.
The point metrics that we consider are the $L_\infty$ and $L_2$~metrics.
The problem of computing the \frechet\ distance has been widely investigated~\cite{Bringmann14,BM16,AAKS14}, and in particular Bringmann and Mulzer~\cite{Bringmann14} showed that strongly subquadratic algorithms for the discrete \frechet\ distance are unlikely to exist.

Several hardness of approximation results for both the \nn{} and \klc{} problems are known.
For the \nn{} problem under the discrete \frechet\ distance, no data structure exists requiring $O(n^{2-\eps}\polylog{m})$ preprocessing and $O(n^{1-\eps}\polylog{m})$ query time for $\eps>0$, and achieving an approximation factor of $c<3$, unless the strong exponential time hypothesis fails~\cite{IM04,DKS16}.
In the case of the \klc{} problem under the discrete \frechet\ distance, Driemel et~al. showed that the problem is $\mathsf{NP}$-hard to approximate within a factor of $2-\eps$ when $k$ is part of the input, even if $\ell=2$ and $d=1$.
Furthermore, the problem is $\mathsf{NP}$-hard to approximate within a factor $2-\eps$ when $\ell$ is part of the input, even if $k=2$ and $d=1$, and when $d=2$ the inapproximability bound is $3\sin{\pi/3} \approx 2.598$~\cite{BDGHKLS19}.

However, we are interested in algorithms that can process large inputs, i.e., where $n$ and/or $m$ are large, which suggests that the processing time ought to be near-linear in $nm$ and the query time for \nn{} queries should be near-linear in $m$ only.
The above results imply that algorithms for the \nn{} and \klc{} problems that achieve such running times are not realistic.
Moreover, given that strongly subquadratic algorithms for computing the discrete \frechet\ distance are unlikely to exist, an algorithm that must compute pairwise distances explicitly will incur a roughly $O(m^2)$ running time. 
To circumvent these constraints, we focus on specific important settings: for the \nn{} problem, either the query curve is assumed to be a segment or the input curves are segments; and for the~\klc{} problem the center is a segment and $k=1$, i.e., we focus on the~\otc{} problem.

While these restricted settings are of theoretical interest, they also have a practical motivation when the inputs are trajectories of objects moving through space, such as migrating birds. 
A segment $ab$ can be considered a trip from a starting point $a$ to a destination $b$. 
Given a set of trajectories that travel from point to point in a noisy manner, we may wish to find the trajectory that most closely follows a direct path from $a$ to $b$, which is the \nn{} problem with a segment query.
Conversely, given an input of (directed) segments
and a query trajectory, the \nn{} problem would identify the segment (the simplest possible trajectory, in a sense) that the query trajectory most closely resembles.
In the case of the \otc{} problem, the obtained segment center for an input of trajectories would similarly represent the summary direction of the input, and the radius $r^*$ of the solution would be a measure of the maximum deviation from that direction for the collection.

\paragraph*{Our results.}

We present algorithms for a variety of settings (summarized in the table below) that achieve the desired running time and storage bounds.
Under the $L_\infty$ metric, we give exact algorithms for the \nn{} and \otc{} problems, including under translation, that achieve the roughly linear bounds. 
For the $L_2$ metric, $(1+\eps)$-approximation algorithms with near-linear running times are given for the \nn{} problem, and for the \otc{} problem, an exact algorithm is given whose running time is roughly $O(n^2m^3)$ and whose space requirement is quadratic. (Parentheses point to results under translation.)

\begin{table}[h]
\begin{center}
	\begin{tabular}{|l|p{4cm}|p{4cm}|p{4cm}|}
		    \hline
              & \textbf{Input/query:} & \textbf{Input/query:} & \textbf{Input:} \\
	    & \textbf{$m$-curves/segment} & \textbf{segments/$m$-curve} & \textbf{(1,2)-center} \\
			  \hline
		\textbf{$L_\infty$} & \Cref{sub:linfty:curve} (\Cref{sub:translation:curve}) & \Cref{sub:linfty:seg} (\Cref{sub:translation:seg}) & \Cref{sec:otc:linfty} (\Cref{sec:otc:translation})\\
		    \hline
		\textbf{$L_2$} & 
		\Cref{sub:l2:curve} & 
		\Cref{sub:l2:seg} &
		\Cref{sec:otc:l2} \\
		    \hline
	\end{tabular}
	\end{center}
	\label{table:query_results}
\end{table}


\section{Preliminaries}
\label{sec:preliminaries}

The discrete \frechet\ distance is a measure of similarity between two curves,  defined as follows.  Consider the curves $C = (p_1,\dotsc,p_m)$ and $C' = (q_1,\dotsc,q_{m'})$, viewed as sequences of vertices.  A (monotone) \emph{alignment} of the two curves is a sequence $\tau := \langle(p_{i_1},q_{j_1}),\dotsc,(p_{i_v},q_{j_v})\rangle$  of pairs of vertices, one from each curve, with $(i_1,j_1)=(1,1)$ and $(i_v,j_v) = (m,m')$.  Moreover, for each pair $(i_u,j_u)$, $1 < u \leq v$, one of the following holds:
(i)~$i_u = i_{u - 1}$ and $j_u = j_{u-1} + 1$,
(ii)~$i_u = i_{u-1} + 1$ and $j_u = j_{u-1}$, or
(iii)~$i_u = i_{u-1} + 1$ and $j_u = j_{u-1} + 1$.
The discrete \frechet\ distance is defined as 
\[\dfd^d(C,C')=\min_{\tau \in \T} \max_{(i,j)\in\tau} d(p_i, q_j),\]
with the minimum taken over the set $\T$ of all such alignments $\tau$, and where $d$ denotes the metric used for measuring interpoint distances.

We now give two alternative, equivalent definitions of the discrete \frechet\ distance between a segment $s=ab$ and a polygonal curve $C=(p_1,\dotsc,p_m)$ (we will drop the point metric~$d$ from the notation, where it is clear from the context). Let $C[i,j]:=\{p_i,\dotsc,p_j\}$.
Denote by $B(p,r)$ the ball of radius $r$ centered at $p$, in metric~$d$.
The discrete \frechet\ distance between~$s$ and $C$ is at most $r$, if and only if there exists a partition of $C$ into a prefix $C[1,i]$ and a suffix $C[i+1,m]$, such that $B(a,r)$ contains $C[1,i]$ and $B(b,r)$ contains $C[i+1,m]$.

A second equivalent definition is as follows.
Consider the intersections of balls around the points of $C$. Set $I_i(r)=B(p_1,r) \cap \cdots \cap B(p_i,r)$ and $\overline{I}_i(r)=B(p_{i+1},r) \cap \cdots \cap B(p_m,r)$, for $i=1,\dotsc,m-1$. Then, the discrete \frechet\ distance between $s$ and $C$ is at most $r$, if and only if there exists an index $1\le i\le m-1$ such that $a\in I_i(r)$ and $b\in \overline{I}_i(r)$.

Given a set $\C=\{C_1,\ldots,C_n\}$ of $n$ polygonal curves in the plane, the nearest-neighbor problem for curves is formulated as follows:
\begin{problem}[\nn{}]\label{prb:nnc}
  Preprocess $\C$ into a data structure, which, given a query curve $Q$, returns a curve $C\in\C$ with $\dfd(Q,C)=\min_{C_i \in \C}\dfd(Q,C_i)$.
\end{problem}

We consider two variants of \Cref{prb:nnc}: (i)~when the query curve $Q$ is a segment, and (ii)~when the input $\C$ is a set of segments.

Secondly, we consider a particular case of the \klc{} problem for curves~\cite{DKS16}. %
\begin{problem}[\otc]
	Find a segment $s^*$ that minimizes $\max_{C_i \in \C} \dfd(s,C_i)$, over all segments $s$.  
\end{problem}

\section{\nn{} and \texorpdfstring{$L_\infty$}{L-infinity} metric}
\label{sec:linfty}
When $d$ is the $L_\infty$ metric, each ball $B(p_i,r)$ is a square. Denote by $S(p,d)$ the axis-parallel square of radius $d$ centered at $p$.

Given a curve $C=(p_1,\ldots,p_m)$, let $d_i$, for $i = 1,\ldots,m-1$, be the smallest radius such that $S(p_1,d_i) \cap \cdots \cap S(p_i,d_i) \ne \emptyset$. In other words, $d_i$ is the radius of the smallest enclosing square of $C[1,i]$.
Similarly, let $\overline{d}_i$, for $i = 1, \ldots, m-1$, be the smallest radius such that $S(p_{i+1},\overline{d}_i) \cap \cdots \cap S(p_m,\overline{d}_i) \ne \emptyset$. 

For any $d > d_i$, $S(p_1,d) \cap \cdots \cap S(p_i,d)$ is a rectangle, $R_i=R_i(d)$, defined by four sides of the squares $S(p_1,d),\ldots,S(p_i,d)$, see \Cref{fig:intersection}.
These sides are fixed and do not depend on the specific value of $d$. Furthermore, the left, right, bottom and top sides of $R_i(d)$ are provided by the sides corresponding to the right-, left-, top- and bottom-most vertices in $C[1,i]$, respectively, i.e., the sides corresponding to the vertices defining the bounding box of $C[1,i]$.

\begin{figure}[h]
	\begin{center}
		\includegraphics[scale=1]{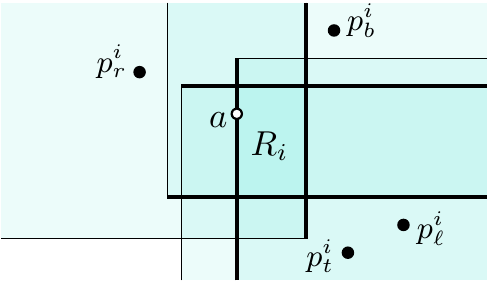}
	\end{center}
	\caption{\label{fig:intersection} The rectangle $R_i=R_i(d)$ and the vertices of the $i$th prefix of~$C$ that define it.}
\end{figure}

Denote by $p^i_{\ell}$ the vertex in the $i$th prefix of $C$ that contributes the left side %
to $R_i(d)$, i.e., the left side of $S(p^i_{\ell},d)$ defines the left side of $R_i(d)$. 
Furthermore, denote by $p^i_r$, $p^i_b$, and~$p^i_t$ the vertices of the $i$th prefix of $C$ that contribute the right, bottom, and top sides to~$R_i(d)$, respectively. 
Similarly, for any $d > \overline{d}_i$, we denote the four vertices of the $i$th suffix of $C$ that contribute the four sides of the rectangle $\overline{R}_i(d) = S(p_{i+1},d) \cap \cdots \cap S(p_m,d)$ by~$\overline{p}^i_{\ell}$, $\overline{p}^i_r$, $\overline{p}^i_b$, and $\overline{p}^i_t$, respectively.  

Finally, we use the notation $R^j_i = R^j_i(d)$ ($\overline{R}^j_i = \overline{R}^j_i(d)$) to refer to the rectangle $R_i = R_i(d)$ ($\overline{R}_i = \overline{R}_i(d)$) of curve $C_j$.

\begin{observation}\label{obs:linfty}
	Let $s=ab$ be a segment, $C$ be a curve, and let $d > 0$. Then, $\dfd(s, C) \le d$ if and only if there exists $i$, $1 \le i \le m-1$, such that $a \in R_i(d)$ and $b \in \overline{R}_{i}(d)$.
\end{observation}

\subsection{Query is a segment}
\label{sub:linfty:curve}
Let $\C = \{C_1,\dotsc,C_n\}$ be the input curves, each of size $m$.
Given a query segment $s=ab$, the task is to find a curve $C \in \C$ such that $\dfd(s,C) = \min_{C' \in \C} \dfd(s,C')$.

\paragraph*{The data structure.}
The data structure is an eight-level search tree. The first level of the data structure is a search tree for the $x$-coordinates of the vertices $p^i_{\ell}$, over all curves~$C \in \C$, corresponding to the $nm$~left sides of the $nm$~rectangles~$R_i(d)$.  The second level corresponds to the $nm$ right sides of the rectangles $R_i(d)$, over all curves $C \in \C$. That is, for each node $u$ in the first level, we construct a search tree for the subset of $x$-coordinates of vertices $p^i_r$ which corresponds to the canonical set of $u$. Levels three and four of the data structure correspond to the bottom and top sides, respectively, of the rectangles $R_i(d)$, over all curves $C \in \C$, and they are constructed using the $y$-coordinates of the vertices $p^i_b$ and the $y$-coordinates of the vertices $p^i_t$, respectively.
The fifth level is constructed as follows. For each node $u$ in the fourth level, we construct a search tree for the subset of $x$-coordinates of vertices $\overline{p}^i_{\ell}$ which corresponds to the canonical set of $u$; that is, if the $y$-coordinate of $p^j_t$ is in $u$'s canonical subset, then the $x$-coordinate of $\overline{p}^{j}_{\ell}$ is in the subset corresponding to $u$'s canonical set.
The bottom four levels correspond to the four sides of the rectangles $\overline{R}_i(d)$ and are built using the $x$-coordinates of the vertices $\overline{p}^i_{\ell}$, the $x$-coordinates of the vertices $\overline{p}^i_r$, the $y$-coordinates of the vertices $\overline{p}^i_b$, and the $y$-coordinates of the vertices $\overline{p}^i_t$, respectively.

\paragraph*{The query algorithm.}
Given a segment $s=ab$ and a distance $d > 0$, we can use our data structure to determine whether there exists a curve $C \in \C$, such that $\dfd(s,C) \le d$. The search in the first and second levels of the data structure is done with $a.x$, the $x$-coordinate of $a$, in the third and fourth levels with $a.y$, in the fifth and sixth levels with $b.x$ and in the last two levels with $b.y$. When searching in the first level, instead of performing a comparison between $a.x$ and the value $v$ that is stored in the current node (which is an $x$-coordinate of some vertex $p^i_{\ell}$), we determine whether $a.x \ge v-d$. Similarly, when searching in the second level, at each node that we visit we determine whether $a.x \le v+d$, where $v$ is the value that is stored in the node, etc.  

Notice that if we store the list of curves that are represented in the canonical subset of each node in the bottom (i.e., eighth) level of the structure, then curves whose distance from $s$ is at most $d$ may also be reported in additional time roughly linear in their number.  

\paragraph*{Finding the closest curve.}
Let $s=ab$ be a segment, let $C$ be the curve in $\C$ that is closest to $s$ and set $d^*=\dfd(s,C)$. Then, there exists $1 \le i \le m-1$, such that $a \in R_i(d^*)$ and~$b \in \overline{R}_{i}(d^*)$. %
Moreover, one of the endpoints $a$ or $b$ lies on the boundary of its rectangle, since, otherwise, we could shrink the rectangles without `losing' the endpoints. Assume without loss of generality that $a$ lies on the left side of $R_i(d^*)$. Then, the difference between the $x$-coordinate of the vertex $p^i_{\ell}$ and $a.x$ is exactly $d^*$. This implies that we can find $d^*$ by performing a binary search in the set of all $x$-coordinates of vertices of curves in $\C$. In each step of the binary search, we need to determine whether $d \ge d^*$, where $d=v-a.x$ and $v$ is the current $x$-coordinate, and our goal is to find the smallest such $d$ for which the answer is still yes. We resolve a comparison by calling our data structure with the appropriate distance $d$. Since we do not know which of the two endpoints, $a$ or $b$, lies on the boundary of its rectangle and on which of its sides, we perform 8 binary searches, where each search returns a candidate distance. Finally, the smallest among these 8 candidate distances is the desired~$d^*$.

In other words, we perform 4 binary searches in the set of all $x$-coordinates of vertices of curves in $\C$. In the first we search for the smallest distance among the distances $d_{\ell}=v-a.x$ for which there exists a curve at distance at most $d_{\ell}$ from~$s$; in the second we search for the smallest distance $d_r=a.x-v$ for which there exists a curve at distance at most $d_r$ from $s$; in the third we search for the smallest distance $\overline{d}_{\ell}=v-b.x$ for which there exists a curve at distance at most $\overline{d}_{\ell}$ from $s$; and in the fourth we search for the smallest distance $\overline{d}_r=b.x-v$ for which there exists a curve at distance at most $\overline{d}_r$ from $s$. We also perform 4 binary searches in the set of all $y$-coordinates of vertices of curves in $\C$, obtaining the candidates $d_b$, $d_t$, $\overline{d}_b$, and $\overline{d}_t$. We then return the distance $d^* = \min \{d_{\ell},d_r,\overline{d}_{\ell},\overline{d}_r,d_b,d_u,\overline{d}_b,\overline{d}_u\}$.

\begin{theorem}
	Given a set $\C$ of $n$ curves, each of size $m$, one can construct a search structure of size $O(nm \log^7 (nm))$ for segment nearest-curve queries. Given a query segment $s$, one can find in $O(\log^8 (nm))$ time the curve $C \in \C$ and distance $d^*$ such that $\dfd(s,C) = d^*$ and~$d^* \le \dfd(s,C')$ for all $C' \in \C$, under the $L_\infty$ metric.
\end{theorem}

\subsection{Input is a set of segments}
\label{sub:linfty:seg}
Let $\S = \{s_1,\dotsc,s_n\}$ be the input set of segments. %
Given a query curve $Q=(p_1,\dotsc,p_m)$, the task is to find a segment $s = ab \in \S$ such that $\dfd(Q,s) = \min_{s' \in \S} \dfd(Q,s')$, after suitably preprocessing $\S$.  We use an overall approach similar to that used in \Cref{sub:linfty:curve}, however the details of the implementation of the data structure and algorithm differ.

\paragraph*{The data structure.}
Preprocess the input $\S$ into a four-level search structure~$\T$ consisting of a two-dimensional range tree containing the endpoints $a$, and where the associated structure for each node in the second level of the tree is another two-dimensional range tree containing the endpoints $b$ corresponding to the points in the canonical subset of the node.

This structure answers queries consisting of a pair of two-dimensional ranges (i.e., rectangles) $(R,\overline{R})$ and returns all segments $s = ab$ such that $a \in R$ and $b \in \overline{R}$.
The preprocessing time for the structure is $O(n \log^4{n})$, and the storage is $O(n \log^3{n})$. %
Querying the structure with two rectangles requires $O(\log^3{n})$ time, by applying fractional cascading~\cite{Willard85}.

\paragraph*{The query algorithm.}
Consider the decision version of the problem where, given a query curve $Q$ and a distance $d$, the objective is to determine if there exists a segment $s\in\S$ with~$\dfd(s,Q)\le d$. 
\Cref{obs:linfty} implies that it is sufficient to query the search structure $\T$ with the pair of rectangles $(R_i(d),\overline{R}_i(d))$ of the curve $Q$, for all $1\le i\le m-1$.
If $\T$ returns at least one segment for any of the partitions, then this segment is within distance $d$ of $Q$.

As we traverse the curve $Q$ left-to-right, the bounding box of $Q[1,i]$ can be computed at constant incremental cost.  For a fixed~$d>0$, each rectangle~$R_i(d)$ can be constructed from the corresponding bounding box in constant time.  
Rectangle~$\overline{R}_i(d)$ can be handled similarly by a reverse traversal. Hence all the rectangles can be computed in time $O(m)$, for a fixed~$d$.
Each pair of rectangles requires a query in~$\T$, and thus the time required to answer the decision problem is $O(m\log^3{n})$.

\paragraph*{Finding the closest segment.}

In order to determine the nearest segment $s$ to $Q$, we claim, using an argument similar to that in \Cref{sub:linfty:curve}, for a segment $s = ab$ of distance $d^*$ from $Q$ that either $a$ lies on the boundary of $R_i(d^*)$ or $b$ lies on the boundary of $\overline{R}_i(d^*)$ for some~$1\leq i < m$.
Thus, in order to determine the value of $d^*$ it suffices to search over all $8m$ pairs of rectangles where either $a$ or $b$ lies on one of the eight sides of the obtained query rectangles.
The sorted list of candidate values of $d$ for each side can be computed in $O(n)$ time from a sorted list of the corresponding $x$- or $y$-coordinates of $a$ or $b$.
The smallest value of $d$ for each side is then obtained by a binary search of the sorted list of candidate values. 
For each of the $O(\log{n})$ evaluated values $d$, a call to $\T$ decides on the existence of a segment within $d$ of $Q$.

\begin{theorem}
    Given an input $\S$ of $n$ segments, a search structure can be preprocessed in~$O(n\log^4{n})$ time and requiring $O(n\log^3{n})$ storage that can answer the following.
    For a query curve $Q$ of $m$ vertices, find the segment $s^* \in \S$ and distance $d^*$ such that $\dfd(Q,s^*) = d^*$ and $\dfd(Q,s) \geq d^*$ for all $s \in \S$ under the $L_\infty$ metric.
    The time to answer the query is~$O(m\log^4{n})$.
\end{theorem}

\section{\nn{} under translation and \texorpdfstring{$L_{\infty}$}{L-infinity} metric}
\label{sec:translation}
An analogous approach yields algorithms with similar running times for the problems under translation.

For a curve $C$ and a translation $t$, let $C_t$ be the curve obtained by translating $C$ by $t$, i.e., by translating each of the vertices of $C$ by $t$.
In this section we study the two problems studied in \Cref{sec:linfty}, assuming the input curves are given up to translation. That is, the distance between the query curve $Q$ and an input curve $C$ is now $\min_t \dfd(Q,C_t)$, where the discrete \frechet\ distance is computed using the $L_\infty$ metric.

\subsection{Query is a segment}
\label{sub:translation:curve}

Let $\C = \{C_1,\dotsc,C_n\}$ be the set of input curves, each of size $m$. We need to preprocess~$\C$ for segment nearest-neighbor queries under translation, that is,
given a query segment $s=ab$,
find the curve $C \in C$ that minimizes $\min_t \dfd(s,C'_t) = \min_t \dfd(s_t,C')$,
where $s_t$ and $C_t$ are the images of $s$ and $C$, respectively, under the translation $t$.
Let $t^*$ be the translation that minimizes $\dfd(s_t, C)$,
and set $d^* = \dfd(s_{t^*}, C)$. Consider the partition of $C=(p_1,\ldots,p_m)$ into prefix $C[1,i]$ and suffix $C[i+1,m]$, such that $a_{t^*} \in R_i(d^*)$ and $b_{t^*} \in \overline{R}_i(d^*)$. The following trivial observation allows us to construct a set of values to which $d^*$ must belong.   

\begin{observation}
	One of the following statements holds: %
	\begin{enumerate}
		\item
		$a_{t^*}$ lies on the left side of $R_i(d^*)$ and $b_{t^*}$ lies on the right side of $\overline{R}_i(d^*)$, or vice versa,
		i.e., $a_{t^*}$ lies on the right side of $R_i(d^*)$ and $b_{t^*}$ lies on the left side of $\overline{R}_i(d^*)$.
		\item
		$a_{t^*}$ lies on the bottom side of $R_i(d^*)$ and $b_{t^*}$ lies on the top side of $\overline{R}_i(d^*)$, or vice versa.
	\end{enumerate}
\end{observation}

Assume without loss of generality that $a.x < b.x$ and $a.y < b.y$ and that the first statement holds. Let $\delta_x = b.x - a.x$ denote the $x$-span of $s$, and let $\delta_y$ denote the $y$-span of $s$. Then, either (i) $(\overline{p}^i_r.x + d^*) - (p^i_l.x - d^*) = \delta_x$, or (ii) $(\overline{p}^i_l.x - d^*) - (p^i_r.x + d^*) = \delta_x$, where as before $p^i_l$ ($p^i_r$) is the vertex of $C$ which determines the left (right) side of $R_i$ and $\overline{p}^i_l$ ($\overline{p}^i_r$) is the vertex of $C$ which determines the left (right) side of $\overline{R}_i$. That is, either (i)~$d^* = \frac{\delta_x - (\overline{p}^i_r.x - p^i_l.x)}{2}$, or (ii)~$d^* = \frac{(\overline{p}^i_l.x - p^i_r.x) - \delta_x}{2}$. 

\paragraph*{The data structure.}
Consider the decision version of the problem: Given $d$, is there a curve in~$\C$ whose distance from $s$ under translation is at most $d$? 
We now present a five-level data structure to answer such decision queries.
We continue to assume that $a.x < b.x$ and $a.y < b.y$.  For a curve $C_j$, let $d^j_i$ ($\overline{d}^j_i$) be the smallest radius such that
$R^j_i$ ($\overline{R}^j_i$) is non-empty, and set $r^j_i = \max\{d^j_i, \overline{d}^j_i\}$. The top level of the structure is simply a binary search tree on the $n(m-1)$ values $r^j_i$; it serves to locate 
the pairs $(R^j_i(d), \overline{R}^j_i(d))$ in which both rectangles are non-empty. The role of the remaining four levels is to filter the set of relevant pairs, so that at the bottom level we remain with those pairs for which $s$ can be translated so that $a$ is in the first rectangle and $b$ is in the second. 

For each node $v$ in the top level tree, we construct a search tree over the values $\overline{p}^i_r.x - p^i_l.x$ corresponding to the pairs in the canonical subset of $v$. These trees constitute the second level of the structure. The third-level trees are search trees over the values $\overline{p}^i_l.x - p^i_r.x$, the fourth-level ones~--- over the values $\overline{p}^i_t.y - p^i_b.y$, and finally the fifth-level ones~--- over the  values $\overline{p}^i_b.y - p^i_t.y$.

\paragraph*{The query algorithm.}
Given a query segment $s=ab$ (with $a.x < b.x$ and $a.y < b.y$) and $d > 0$, we employ our data structure to answer the decision problem. In the top level, we select all pairs $(R^j_i, \overline{R}^j_i)$ satisfying $r^j_i \le d$. Of these pairs, in the second level, we select all pairs satisfying $\overline{p}^i_r.x - p^i_l.x \ge \delta_x - 2d$. In the third level, we select all pairs satisfying $\overline{p}^i_l.x - p^i_r.x \le \delta_x + 2d$. Similarly, in the fourth level, we select all pairs satisfying $\overline{p}^i_t.y - p^i_b.y \ge \delta_y - 2d$, and in the fifth level, we select all pairs satisfying $\overline{p}^i_b.y - p^i_t.y \le \delta_y + 2d$. At this point, if our current set of pairs is non-empty, we return \textsc{yes}, otherwise, we return \textsc{no}.

To find the nearest curve $C$ and the corresponding distance $d^*$, we proceed as follows, utilizing the observation above. We perform a binary search over the $O(nm)$ values of the form $\overline{p}^i_r.x - p^i_l.x$ to find the largest value for which the decision algorithm returns \textsc{yes} on $d = \frac{\delta_x - (\overline{p}^i_r.x - p^i_l.x)}{2}$. (We only consider the values $\overline{p}^i_r.x - p^i_l.x$ that are smaller than $\delta_x$.) Similarly, we perform a binary search over the values $\overline{p}^i_t.y - p^i_b.y$ to find the largest value for which the decision algorithm returns \textsc{yes} on $d = \frac{\delta_y - (\overline{p}^i_t.y - p^i_b.y)}{2}$. We perform two more binary searches; one over the values $\overline{p}^i_l.x - p^i_r.x$ to find the smallest value for which the decision algorithms returns \textsc{yes} on $d = \frac{(\overline{p}^i_l.x - p^i_r.x) - \delta_x}{2}$, and one over the values $\overline{p}^i_b.y - p^i_t.y$. Finally, we return the smallest $d$ for which the decision algorithm has returned \textsc{yes}.  

Our data structure was designed for the case where $b$ lies to the right and above $a$. Symmetric data structures for the other three cases are also needed. The following theorem summarizes the main result of this section.

\begin{restatable}{theorem}{thmtranslationquery}
	Given a set $\C$ of $n$ curves, each of size $m$, one can construct a search structure of size~$O(nm \log^4 (nm))$, such that, given a query segment $s$, one can find in~$O(\log^6(nm))$~time the curve $C \in C$ nearest to $s$ under translation,
	that is the curve minimizing $\min_t \dfd(s_t, C')$,
	where the discrete \frechet\ distance is computed using the $L_\infty$ metric. 
\end{restatable}

\subsection{Input is a set of segments}
\label{sub:translation:seg}
Let $\S = \{s_1,\dotsc,s_n\}$ be the input set of segments, with $s_i = a_{i}b_{i}$.
We need to preprocess~$\S$ for nearest-neighbor queries under translation, that is,
given a query curve $Q=(p_1,\dots,p_m)$, find the segment $s = ab \in \S$ that minimizes
$\min_t \dfd(Q,s'_t) = \min_t \dfd(Q_t,s')$. 
Since translations are allowed, without loss of generality
we can assume that the first point of all the segments is the origin. In other words,
the input is converted to a two-dimensional point set $\C=\{c_i = b_i-a_i  \mid a_{i}b_{i}\in \S\}$.

The idea is to find the nearest segment corresponding to each of the $m-1$ partitions of the query.
Let $s=ab$ be any segment and $d$ some radius. The following observation holds for any partition of $Q$ into $Q[1, i]$ and $Q[i+1,m]$, where $R_i^\oplus(d) = ( -R_i(d)) \oplus \overline{R}_i(d)$ and $\oplus$ is the Minkowski sum operator, see \Cref{fig:minkowski}.

\begin{figure}[h]
	\begin{center}
		\includegraphics[scale=1.1]{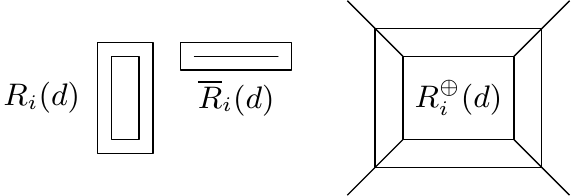}
	\end{center}
	\caption{\label{fig:minkowski} The rectangle $R_i^\oplus(d)$, as $d$ increases.}
\end{figure}

\begin{observation} \label{obs:seg:location}
	There exists a translation $t$ such that $a_t \in R_i(d)$ and $b_t \in \overline{R}_i(d)$ if and only if $c=b-a \in R_i^\oplus(d)$.
\end{observation}

Based on this observation segment $ab$ is within distance $d$ of $Q$ under translation, if for some $i$, $R_i^\oplus(d)$ contains the point $c=b-a$, which means translations can be handled implicitly.

\paragraph*{The data structure.}
According to \Cref{obs:seg:location}, a data structure is required to answer the following question: Given a partition of $Q$ into prefix $Q[1,i]$ and suffix $Q[i+1,m]$, what is the smallest radius $d^*$ so that $R_i^\oplus(d^*)$ contains some $c_j \in \C$?
The smallest radius $d^\prime$ where both $R_i(d^\prime)$ and $\overline{R}_i(d^\prime)$---and hence $R_i^\oplus(d^\prime)$---are nonempty can be determined in linear time. This value which depends on $i$ is a lower bound on $d^*$.

Since $-R_i(d^\prime)$ and $\overline{R}_i(d^\prime)$ are both axis-aligned rectangles (segments or points in special cases), their Minkowski sum, $R_i^\oplus(d^\prime)$, is also a possibly degenerate axis-aligned rectangle. %
If this rectangle contains some point $c_j\in \C$, then $s_j$ is the nearest segment with respect to this partition and the optimal distance is $d^\prime$. If it contains more than one point from $\C$, then all the corresponding segments are equidistant from the query and each of them can be reported as the nearest neighbor corresponding to this partition. The data structure needed here is a two-dimensional range tree on $\C$.

If $R_i^\oplus(d^\prime) \cap \C$ is empty, then we need to find the smallest radius $d^*$ so that $R_i^\oplus(d^*)$ contains some $c_j$.
For any distance $d>d^\prime$, $R_i^\oplus(d)$ is a rectangle concentric with $R_i^\oplus(d^\prime)$ but whose edges are longer by an additive amount of $4(d-d^\prime)$. 

As $d$ increases, the four edges of the rectangle sweep through 4 non-overlapping regions in the plane, so any point in the plane that gets covered by $R_i^\oplus(d)$, first appears on some edge.
We divide this problem into 4 sub-problems based on the edge that the optimal $c_j$ might appear on. Below, we solve the sub-problem for the right edge of the rectangle: Given a partition of $Q$ into prefix $Q[1,i]$ and suffix $Q[i+1,m]$, what is the smallest radius $d_r^*$ so that the right edge of $R_i^\oplus(d_r^*)$ contains some $c_j$? All other sub-problems are solved symmetrically.

Any point $c_j$ that appears on the right edge belongs to the intersection of three half-planes:
\begin{enumerate}
	\item On or below the line of slope $+1$ passing through the top-right corner of the rectangle $R_i^\oplus(d')$.
	\item On or above the line of slope $-1$ passing through the bottom-right corner of $R_i^\oplus(d')$.
	\item To the right of the line through the right edge of $R_i^\oplus(d')$.
\end{enumerate}

The first point in this region %
swept by the right edge of the growing rectangle $R_i^\oplus(d)$ is the one with
the smallest $x$-coordinate. This point can be located using a three-dimensional range tree on $\C$.

\paragraph*{The query algorithm.}

Given a query curve $Q=(p_1,\dotsc,p_m)$, the nearest segment under translation can be determined by using the data structure to find the nearest segment---and its distance from $Q$---for each of the $m-1$ partitions and selecting the segment whose distance is smallest.

As stated in \Cref{sub:linfty:seg}, all $O(m)$ bounding boxes can be computed in $O(m)$ total time. For a particular partition, knowing the two bounding boxes, one can determine the smallest radius $d^\prime$ where $R_i^\oplus(d^\prime)$ is nonempty in constant time.
Now the two-dimensional range tree on $\C$ is used to search for points inside $R_i^\oplus(d^\prime)$. If the data structure returns some point $c\in \C$, then the segment corresponding to $c$ is the nearest segment under translation. Otherwise, one has to do four three-level range searches in the second data structure and compare the results to find the nearest segment. This is the most expensive step which takes $O(\log^2{n})$ time using fractional cascading~\cite{Willard85}. The following theorem summarizes the main result of this section.

\begin{restatable}{theorem}{thmtranslationinput}
	Given a set $\S$ of $n$ segments, one can construct a search structure of size~$O(n \log^2 {n})$, so that, given a query curve $Q$ of size $m$, one can find in~$O(m\log^2{n})$ time the segment $s \in \S$ nearest to $Q$ under translation,
	that is the segment minimizing $\min_t \dfd(Q, s'_t)$,
	where the discrete \frechet\ distance is computed using the $L_\infty$ metric. 
\end{restatable}

\section{\nn{} and \texorpdfstring{$L_2$}{L-2} metric}
\label{sec:l2}

In this section, we present algorithms for approximate nearest-neighbor search under the discrete \frechet\ distance using $L_2$. 
Notice that the algorithms from \Cref{sec:linfty} for the  $L_\infty$ version of the problem, already give $\sqrt{2}$-approximation algorithms for the $L_2$ version.  Next, we provide $(1+\eps)$-approximation algorithms.

\subsection{Query is a segment}
\label{sub:l2:curve}
Let $\C=\{C_1,\ldots,C_n\}$ be a set of $n$ polygonal curves in the plane. The $(1+\eps)$-approximate nearest-neighbor problem is defined as follows: Given $0<\eps\le 1$, preprocess $\C$ into a data structure supporting queries of the following type: given a query segment $s$, return a curve~$C' \in \C$, such that $\dfd(s,C') \le (1+\eps)\dfd(s,C)$, where $C$ is the curve in $\C$ closest to $s$.

Here we provide a data structure for the $(1+\eps,r)$-approximate nearest-neighbor problem, defined as:
Given a parameter $r$ and $0<\eps\le 1$, preprocess $\C$ into a data structure supporting queries of the following type: given a query segment $s$, if there exists a curve $C_i\in \C$ such that $\dfd(s,C_i)\le r$, then return a curve $C_j\in \C$ such that $\dfd(s,C_j)\le (1+\eps)r$.

There exists a reduction from the $(1+\eps)$-approximate nearest-neighbor problem to the $(1+\eps,r)$-approximate nearest-neighbor problem~\cite{Indyk00}, at the cost of an additional logarithmic factor in the query time. 

\paragraph*{An exponential grid.}
Given a point $p\in\reals{}^{2}$, a parameter $0<\eps\le1$, and an interval $[\alpha,\beta]\subseteq\reals{}$, we can construct the following exponential grid $G(p)$ around $p$, which is a slightly different version of the exponential grid presented in~\cite{Driemel13}:

Consider the series of axis-parallel squares $S_i$ centered at $p$ and of side lengths $\lambda_i=2^i\alpha$, for $i=1,\dotsc,\lceil \log(\beta/\alpha)\rceil$. Inside each region $S_i\setminus S_{i-1}$ (for $i>1$), construct a grid $G_i$ of side length $\frac{\eps\lambda_i}{2\sqrt{2}}$.
The total number of grid cells is at most 
\[
1 + \sum_{i=2}^{\lceil \log(\beta/\alpha)\rceil}\big(\lambda_i / \frac{\eps\lambda_i}{2\sqrt{2}}\big)^2
=O((1/\eps)^2\lceil \log(\beta/\alpha)\rceil).
\]

Given a point $q\in\reals{}^2$ such that $\alpha\le\| q-p \|\le\beta$, 
let $i$ be the smallest index such that $q\in S_i$. If $q$ is in $S_1$, then $\|q-p\|\le\sqrt{2}\alpha$.
Else, we have $i>1$. Let $g$ be the grid cell of $G_i$ that contains $q$, and denote by $c_g$ the center point of $g$. 
So we have $\|q-c_g\|\le \frac{\sqrt{2}}{2}\frac{\eps\lambda_i}{2\sqrt{2}}=\frac{\eps}{2} 2^{i-1}\alpha\le \frac{\eps}{2} 2^{\log(\beta/\alpha)}\alpha=\frac{\eps\beta}{2}$.

\paragraph*{A data structure for $(1+\eps,r)$-\ann{}.}
For each curve $C_i=(p^i_1,\ldots,p^i_m) \in \C$, we construct two exponential grids: $G(p^i_1)$ around $p^i_1$ and $G(p^i_m)$ around $p^i_m$, both with the range $[\frac{\eps r}{2\sqrt{2}},\, r]$, as described above.
Now, for each pair of grid cells $(g,h)\in G(p^i_1)\times G(p^i_m)$, let $C(g,h)=C\in\C$ be the curve such that $\dfd(c_gc_h,C)=\min_j\{\dfd(c_gc_h,C_j)\}$. In other words, $C(g,h)$ is the closest input curve to the segment $c_gc_h$.

Let $\G_1$ be the union of the grids $G(p^1_1),G(p^2_1),\dotsc,G(p^n_1)$, and $\G_m$ the union of the grids $G(p^1_m),G(p^2_m),\dotsc,G(p^n_m)$.
The number of grid cells in each grid is $O((1/\eps)^2\lceil \log(r/\frac{\eps r}{2\sqrt{2}})\rceil)=O(\frac{1}{\eps^2}\log(1/\eps))$.
The number of grid cells in $\G_1$ and $\G_m$ is thus $O(n\frac{1}{\eps^2}\log(1/\eps))$. 

The data structure is a four-level segment tree, where each grid cell is represented in the structure by its bottom- and left-edges.
The first level is a segment tree for the horizontal edges of the cells of $\G_1$. The second level corresponds to the vertical edges of the cells of~$\G_1$: for each node $u$ in the first level, a segment tree is constructed for the set of vertical edges that correspond to the horizontal edges in the canonical subset of $u$. That is, if some horizontal edge of a cell in $G(p^i_1)$ is in $u$'s canonical subset, then the vertical edge of the same cell is in the segment tree of the second level associated with $u$.
Levels three and four of the data structure correspond to the horizontal and vertical edges, respectively, of the cells in $\G_m$.

The third level is constructed as follows. For each node $u$ in the second level, we construct a segment tree for the subset of horizontal edges of cells in $\G_m$ which corresponds to the canonical set of $u$; that is, if a vertical edge of $G(p^i_1)$ is in $u$'s canonical subset, then all the horizontal edges of $G(p^i_m)$ are in the subset corresponding to $u$'s canonical set. Thus, the size of the third-level subset is $O(\frac{1}{\eps^2}\log(1/\eps))$ times the size of the second-level subset.

Each node of the forth level corresponds to a subset of pairs of grid cells from the set $\bigcup\limits_{i=1}^{n} (G(p^i_1)\times G(p^i_m))$. In each such node $u$ we store the curve $C(g,h)$ such that $(g,h)$ is the pair in $u$'s corresponding set for which $\dfd(c_g c_h,C(g,h))$ is minimum.

Given a query segment $s=ab$, we can obtain all pairs of grid cells $(g,h)\in\bigcup\limits_{i=1}^{n} (G(p^i_1)\times G(p^i_m))$, such that $a\in g$ and $b\in h$, as a collection of $O(\log^4(\frac{n}{\eps}))$ canonical sets in $O(\log^4(\frac{n}{\eps}))$ time. Then, we can find, within the same time bound, the pair of cells $g,h$ among them for which $\dfd(c_g c_h,C(g,h))$ is minimum.
The space required is $O(n\frac{1}{\eps^4}\log^4(\frac{n}{\eps}))$.

\paragraph*{The query algorithm.}
Given a query segment $s=ab$,
let $p,q$ be the pair of cell center points returned when querying the data structure with $s$, and let $C_j\in\C$ be the closest curve to $pq$.
We show that if there exists a curve $C_i\in \C$ with $\dfd(ab,C_i)\le r$, then $\dfd(ab,C_j)\le (1+\eps) r$.

Since $\dfd(ab,C_i)\le r$, it holds that $\dfd(ab,p^i_1p^i_m)\le r$, and thus there exists a pair of grid cells $g\in G(p^i_1)$ and $h\in G(p^i_m)$ such that $a\in g$ and $b\in h$. The data structure returns~$p,q$, so we have $\dfd(pq,C_j)\le \dfd(c_g c_h,C_i)$\ \ (1). The properties of the exponential grids $G(p^i_1)$ and $G(p^i_m)$ guarantee that $\|a-c_g\|, \|b-c_h\| \le \max\{\sqrt{2}\alpha,\frac{\eps\beta}{2}\}=\frac{\eps}{2}r$. Therefore, $\dfd(c_g c_h,ab)\le \frac{\eps}{2} r$\ \ (2), and, similarly, $\dfd(pq,ab)\le \frac{\eps}{2} r$\ \ (3). 
By the triangle inequality and Equation~(2), 
$\dfd(c_g c_h,C_i)\le \dfd(c_g c_h,ab)+\dfd(ab,C_i)\le (1+\frac{\eps}{2})r$\ \ (4).
Finally, by the triangle inequality and Equations (1), (3) and (4), 
\begin{eqnarray*}
	\dfd(ab,C_j)\le \dfd(ab,pq)+\dfd(pq,C_j) 
	\le\dfd(ab,pq)+\dfd(c_g c_h,C_i) \\
	\le\frac{\eps}{2} r+(1+\frac{\eps}{2})r=(1+\eps)r\,.
\end{eqnarray*}

\begin{restatable}{theorem}{thmnnssegment}
	Given a set $\C$ of $n$ curves, each of size $m$, and $0 < \eps \le 1$, one can construct a search structure of size $O(\frac{n}{\eps^4}\log^4(\frac{n}{\eps}))$ for approximate segment nearest-neighbor queries. Given a query segment $s$, one can find in $O(\log^5(\frac{n}{\eps}))$ time a curve $C' \in \C$ such that $\dfd(s,C') \le (1+\eps)\dfd(s,C)$, under the $L_2$ metric, where $C$ is the curve in $\C$ closest to $s$.
\end{restatable}

\subsection{Input is a set of segments}
\label{sub:l2:seg}

In \Cref{sub:linfty:seg}, we presented an exact algorithm for the problem under $L_\infty$, in which we compute the intersections of the squares of radius $d$ around the vertices of the query curve, and use a two-level data structure for rectangle-pair queries.

To achieve an approximation factor of $(1+\eps)$ for the problem under $L_2$, we can use the same approach, except that instead of squares we use regular $k$-gons. 
Given a query curve $Q=(p_1,\dotsc,p_m)$, the intersections of the regular $k$-gons of radius $d$ around the vertices of $Q$ are polygons with at most $k$ edges, defined by at most $k$ sides of the regular $k$-gons. The orientations of the edges of the intersections are fixed, and thus we can construct a two-level data structure for $k$-gon-pair queries, where each level consists of $k$ inner levels, one for each possible orientation. The size of such a data structure is thus $O(n\log^{2k}n)$.

Given a parameter $\eps$, we pick $k=O(\frac{1}{\sqrt{\eps}})$, so that the approximation factor is $(1+\eps)$, the space complexity is $O(n\log^{O(\frac{1}{\sqrt{\eps}})}n)$ and the query time is $O(m \log^{O(\frac{1}{\sqrt{\eps}})}n)$.

\begin{theorem}
	Given an input $\S$ of $n$ segments, and $0 < \eps \le 1$, one can construct a search structure of size $O(n\log^{O(\frac{1}{\sqrt{\eps}})}n)$ for approximate segment nearest-neighbor queries. Given a query curve $Q$ of size $m$, one can find in $O(m \log^{O(\frac{1}{\sqrt{\eps}})}n)$ time a segment $s' \in \S$ such that $\dfd(s',Q) \le (1+\eps)\dfd(s,Q)$, under the $L_2$ metric, where $s$ is the segment in $\S$ closest to $Q$.
\end{theorem}

\section{\otc}
 \label{sec:otc}

The objective of the \otc{} problem is to find a segment $s$ such that 
$\max_{C_i \in \C} \dfd(s,C_i)$ is minimized.
This can be reformulated equivalenly as:
Find a pair of balls $(B,\overline{B})$, such that (i)~for each curve $C \in \C$, there exists a partition at $1\leq i < m$ of $C$ into prefix $C[1,i]$ and suffix $C[i+1,m]$, with $C[1,i] \subseteq B$ and $C[i+1,m] \subseteq \overline{B}$, and (ii)~the radius of the larger ball is minimized.

\subsection{\otc{} and \texorpdfstring{$L_{\infty}$}{L-infinity} metric}
\label{sec:otc:linfty}

An optimal solution to the \otc{} problem under the $L_\infty$ metric
is a pair of squares $(S,\overline{S})$, where $S$ contains all the prefix vertices and $\overline{S}$ contains all the suffix vertices. 
Assume that the optimal radius is $r^*$, and that it is determined by~$S$, i.e., the radius of $S$ is $r^*$ and the radius of $\overline{S}$ is at most $r^*$. Then, there must exist two \emph{determining vertices} $p,p'$, belonging to the prefixes of their respective curves, such that $p$ and $p'$ lie on opposite sides of the boundary of $S$. Clearly, $||p - p'||_\infty = 2r^*$. Let the positive normal direction of the sides be the \emph{determining direction} of the solution. 

Let $R$ be the axis-aligned bounding rectangle of $C_1\cup\dots\cup C_n$, and denote by $e_\ell$, $e_r$, $e_t$, and $e_b$ the left, right, top, and bottom edges of $R$, respectively.

\begin{restatable}{lemma}{lemmaotcextremal}
  \label{lem:otc:extremal}
      At least one of $p,p'$ must lie on the boundary of $R$.
\end{restatable}
\begin{proof}
	Assume that the determining direction is the positive $x$-direction, and that neither $p$ nor $p'$ lies on the boundary of~$R$. Thus, there must exist a pair of vertices $q,q' \in \overline{S}$ with $q.x < p.x$ and $q'.x > p'.x$, which implies that $||q-q'||_\infty > ||p-p'||_\infty = 2r^*$, contradicting the assumption that $p,p'$ are the determining vertices.
\end{proof}

We say that a corner of $S$ (or $\overline{S}$) \emph{coincides} with a corner of $R$ when the corner points are incident, and they are both of the same type, i.e., top-left, bottom-right, etc.

\begin{restatable}{lemma}{lemmaotccorner}
    \label{lem:otc:corner}
    There exists an optimal solution $(S,\overline{S})$ where at least one corner of~$S$ or $\overline{S}$ coincides with a corner of~$R$.  
\end{restatable}

\begin{proof}
	Let $p,p' \in S$ be a pair of determining vertices, and assume, without loss of generality, that $p$ lies on the boundary of $R$. %
	If $p$ is a corner of $R$, then the claim trivially holds.
	Otherwise, $p$ lies in the interior of an edge of $R$, and assume without loss of generality that it lies on $e_\ell$.
	If $S$ contains a vertex on $e_t$, then we can shift $S$ vertically down until its top edge overlaps $e_t$. Else, if it contains a vertex on $e_b$, then we can shift $S$ up until its bottom edge overlaps~$e_b$.  In both cases, the lemma conclusion holds.
	
	If $S$ does not contain any vertex from $e_t$ or $e_b$, then clearly $\overline{S}$ must contain vertices $q \in e_t$ and $q' \in e_b$ with $||q-q'||_\infty\le 2r^*$. Therefore, $S$ intersects $e_b$ or $e_t$ (or both), and can be shifted vertically until its boundary overlaps $e_b$ or $e_t$, as desired.
	
	A symmetric argument can be made when $p$ and $p'$ are suffix vertices, i.e., $p,p' \in \overline{S}$.
\end{proof}

\Cref{lem:otc:corner} implies that for a given input $\C$ where the determining vertices are in $S$, there must exist an optimal solution where $S$ is positioned so that one of its corners coincides with a corner of the bounding rectangle, and that one of the determining vertices is on the boundary of $R$.
The optimal solution can thus be found by testing all possible candidate squares that satisfy these properties and returning the valid solution that yields the smallest radius.
The algorithm presented in the sequel will compute the radius $r^*$ of an optimal solution $(S^*,\overline{S^*})$ such that $r^*$ is determined by the prefix square $S^*$, see \Cref{fig:center-bb}. The solution where $r^*$ is determined by $\overline{S^*}$ can be computed in a symmetric manner.

\begin{figure}[h]
	\begin{center}
		\includegraphics[scale=0.9]{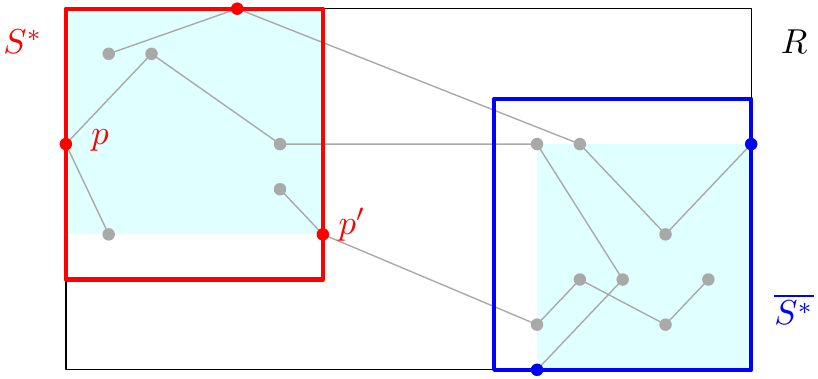}
	\end{center}
	\caption{\label{fig:center-bb} The optimal solution is characterized by a pair of points $p$, $p'$ lying on the boundary of $S^*$, and a corner of $S^*$ coincides with a corner of $R$.}
\end{figure}

For each corner $v$ of the bounding rectangle $R$, we sort the $(m-2)n$ vertices in $C_1 \cup\dots\cup C_n$ that are not endpoints---the initial vertex of each curve must always be contained in the prefix, and the final vertex in the suffix---by their $L_\infty$ distance from $v$. 
Each vertex $p$ in this ordering is associated with a square $S$ of radius $||v-p||_\infty/2$, coinciding with $R$ at corner $v$.

A sequential pass is made over the vertices, and their respective squares~$S$, and for each~$S$ we compute the radius of $S$ and $\overline{S}$ using the following data structures.
We maintain a balanced binary tree $T_C$ for each curve $C\in\C$, where the leaves of $T_C$ correspond to the vertices of $C$, in order.  
Each node of the tree contains a single bit: The bit at a leaf node corresponding to vertex $p_j$ indicates whether $p_j\in S$, where $S$ is the current square. 
The value of the bit at a leaf of $T_C$ can be updated in $O(\log{m})$ time.
The bit of an internal node is 1 if and only if all the bits in the leaves of its subtree are 1, and thus the longest prefix of $C$ can be determined in $O(\log{m})$ time. %
At each step in the pass, the radius of $\overline{S}$ must also be computed, and this is obtained by determining the bounding box of the suffix vertices. 
Thus, two balanced binary trees are maintained: $\overline{T}_x$ contains a leaf for each of the suffix vertices ordered by their $x$-coordinate; and $\overline{T}_y$ where the leaves are ordered by the $y$-coordinate.
The extremal vertices that determine the bounding box can be determined in $O(\log{mn})$ time.
Finally, the current optimal squares $S^*$ and $\overline{S^*}$, and the radius $r^*$ of $S^*$ are persisted. 

The trees $T_{C_1},\ldots,T_{C_n}$ are constructed with all bits initialized to 0, except for the bit corresponding to the initial vertex in each tree which is set to 1, taking $O(nm)$ time in total. 
$\overline{T}_x$ and $\overline{T}_y$ are initialized to contain all non-initial vertices in $O(mn\log{mn})$ time.
The optimal square $S^*$ containing all the initial vertices is computed, and $\overline{S^*}$ is set to contain the remaining vertices.
The optimal radius~$r^*$ is the larger of the radii induced by $S^*$ and $\overline{S^*}$.

At the step in the pass for vertex $p$ of curve $C_j$ whose associated square is~$S$, the leaf of~$T_C$ corresponding to $p$ is updated from $0$ to $1$ in $O(\log{m})$ time.
The index $i$ of the longest prefix covered by $S$ can then be determined, also in $O(\log{m})$ time.
The vertices from $C_j$ that are now in the prefix must be deleted from $\overline{T}_x$ and $\overline{T}_y$, and although there may be $O(m)$ of them in any iteration, each will be deleted exactly once, and so the total update time over the entire sequential pass is $O(mn\log{mn})$.
The radius of the square $S$ is $\Vert v - p \Vert_\infty/2$, and the radius of $\overline{S}$ can be computed in $O(\log{mn})$ time as half the larger of $x$- and $y$-extent of the suffix bounding box.
The optimal squares $S^*$, $\overline{S^*}$, and the cost $r^*$ are updated if the radius of $S$ determines the cost, and the radius of $S$ is less than the existing value of $r^*$.

Finally, we return the optimal pair of squares $(S^*,\overline{S^*})$ with the minimal cost~$r^*$.

\begin{theorem}
    Given a set of curves $\C$ as input, an optimal solution to the \otc\ problem using the discrete \frechet\ distance under the $L_\infty$ metric can be computed in time $O(mn\log{mn})$ using $O(mn)$ storage.
\end{theorem}

\subsection{\otc\ under translation and \texorpdfstring{$L_\infty$}{L-infinity} metric}
\label{sec:otc:translation}

The \otc{} problem under translation and the $L_{\infty}$ metric can be solved using a similar approach.
The objective is to find a segment $s^*$ that minimizes the maximum discrete \frechet\ distance under $L_\infty$ between $s^*$ and the input curves whose locations are fixed only up to translation.
A solution %
will be a pair of squares $(S,\overline{S})$ of equal size and whose radius $r^*$ is minimal, such that, for each $C \in \C$, there exists a translation $t$ and a partition index $i$ where $C_t[1,i] \subset S$ and $C_t[i+i,m] \subset \overline{S}$.
Clearly, an optimal solution will not be unique as the curves can be uniformly translated to obtain an equivalent solution, and moreover, in general there is freedom to translate either square in the direction of at least one of the $x$- or $y$-axes.

Let $\width(C)$ %
be the $x$-extent of the curve $C$ %
and $\height(C)$ be the $y$-extent. %
Let $R$ be the closed rectangle whose bottom-left corner lies at the origin and whose top-right corner is located at $(\delta_x^*, \delta_y^*)$ where $\delta_x^* := \max_{C \in \C} \width(C)$ and $\delta_y^* := \max_{C \in \C} \height(C)$.
Furthermore, let $w_\ell$ and $w_r$ be the left- and right-most vertices in a curve with $x$-span $\delta_x^*$, and let $w_t$ and $w_b$ be the top- and bottom-most vertices in a curve with $y$-span $\delta_y^*$.
Clearly, all curves in $\C$ can be translated to be contained within $R$, and for all such sets of curves under translation, the extremal vertices $w_t$, $w_b$, $w_\ell$ and $w_r$ each must lie on the corresponding side of $R$. %
We claim that if a solution exists with radius $r^*$, then an equivalent solution $(S,\overline{S})$ can be obtained using the same partition of each curve, where $S$ and $\overline{S}$ are placed at opposite corners of $R$. %

\begin{lemma}
	\label{lem:otc:trans:bb}
	Given a set $\C$ of $n$ curves, if there exists a solution of radius $r^*$ to the problem, then there also exists a solution $(S,\overline{S})$ of radius $r^*$ where a corner of $S$ and a corner of $\overline{S}$ coincide with opposite corners of the rectangle $R$. %
\end{lemma}

\begin{proof}
	Let $(S',\overline{S'})$ be a solution of radius $r^*$ where all the curves under translation are not necessarily contained in $R$, and the corners of $S'$ and $\overline{S'}$ do not coincide with the corners of $R$.
	The proof is constructive: The coordinate system is defined such that prefix square $S'$ is positioned so that its corner coincides with the appropriate corner of $R$ ensuring that $S' \equiv S$, and we define a continuous family of squares $\overline{S}(\lambda)$ parameterized on $\lambda \in [0,1]$ where $\overline{S}(0) = \overline{S'}$ and $\overline{S}(1) = \overline{S}$, such that $\overline{S}$ coincides with the opposite corner of $R$. 
	This family traces a translation of $\overline{S}(\lambda)$, first in the $x$-direction and then in the $y$-direction, and we show that the prefix and suffix of each curve---possibly under translation---remain within $S$ and $\overline{S}(\lambda)$, and thus the solution remains valid.
	
	We prove this for the case where the top-right corner $v$ of $S'$ is \emph{below-left} the top-right corner $\overline{v}$ of $\overline{S'}$, i.e., $v.x \leq \overline{v}.x$ and $v.y \leq \overline{v}.y$.
	In the sequel we will show that an equivalent solution $(S,\overline{S})$ exists where the bottom-left corner of $S$ lies at the origin and the top-right corner of $\overline{S}$ lies at $(\delta_x^*,\delta_y^*)$ as required by the claim in the lemma.
	A symmetric argument exists for the other cases where $v$'s position relative to $\overline{c}$ is \emph{above-left}, \emph{below-right} and \emph{below-left}.
	
	First, observe that $\overline{v}.x \geq \delta_x^*$, as either $w_r$ is a vertex in a prefix of some curve and thus $\delta_x^* \leq v.x \leq \overline{v}.x$, or $w_r$ is a vertex in a suffix and thus $\delta_x^* \leq \overline{v}.x$.
	A similar argument proves that $\overline{v}.y \geq \delta_y^*$, and thus $\overline{S}(\lambda)$ will move to the left until the $x$-coordinate of the right edge of $\overline{S}$ is $\delta_x^*$ and then down under the continuous translation to $\overline{S}$, i.e., the $y$-coordinate of the top edge of $\overline{S}$ is $\delta_y^*$.
	
	Consider the validity of the solution $(S,\overline{S}(\lambda))$ as the suffix square moves leftwards. 
	If there are no suffix vertices on the right edge of square $\overline{S}(\lambda)$ then it can be translated to the left and remain a valid solution, until such time as some suffix vertex $\overline{p}$ of curve $C$ lies on the right edge.
	Subsequently, $C$ is translated together with $\overline{S}(\lambda)$, and thus the suffix vertices of $C$ continue to be contained in $\overline{S}(\lambda)$.
	For a prefix vertex $p$ of $C$ to move outside $S$ under the translation it must cross the left-side of $S$, however this would imply that $\vert \overline{p}.x - p.x \vert > \overline{p}.x \geq \delta_x^*$, contradicting the fact that $\delta_x^*$ is the maximum extent in the $x$-direction of all curves.
	The same analysis can be applied to the translation of $\overline{S}(\lambda)$ in the downward direction.
	This shows that the continuous family of squares $\overline{S}(\lambda)$ imply a family of optimal solutions $(S,\overline{S}(\lambda))$ to the problem, and in particular $(S,\overline{S})$ is a solution.
\end{proof}

\Cref{lem:otc:trans:bb} implies that an optimal solution of radius $r^*$ exists where $S$ and $\overline{S}$ coincide with opposite corners of $R$.
Next, we consider the properties of such an optimal solution, and show that $r^*$ is determined by two vertices from a \emph{single} curve.
Recall that a pair of vertices %
are \emph{determining vertices} if they lie on opposite sides of one of the squares.
Here, we refine the definition with the condition that the pair both belong to the prefix or suffix of the same curve.
Furthermore, denote a pair of vertices $(p,\overline{p})$, where $p$ is in the prefix and $\overline{p}$ is in the suffix of the same curve, %
as \emph{opposing vertices} if they preclude a smaller pair of squares coincident with the same opposing corners of $R$. %
Assuming that $S$ coincides with the top-left corner of $R$ and $\overline{S}$ with the bottom-right corner, then %
$p$ and $\overline{p}$ are opposing vertices if, either:
(i) $p$ lies on the right edge of $S$ and $\overline{p}$ lies on the left edge of $\overline{S}$; or 
(ii) $p$ lies on bottom edge of $S$ and $\overline{p}$ lies on the top edge of $\overline{S}$.
Symmetrical conditions exist for the cases where $S$ and $\overline{S}$ are coincident with the other three (ordered) pairs of corners.
We claim that the conditions in the following lemma are necessary for a solution.

\begin{lemma}
	\label{lem:otc:trans:optvert}
	Let $(S,\overline{S})$ be an optimal solution of radius $r^*$ such that $S$ and $\overline{S}$ are coincident with opposite corners of $R$, %
	and let $\C' := \{C_t \mid C \in \C\}$ be the set of curves under translation from which $(S,\overline{S})$ was obtained.
	At least one of the following conditions must hold for some curve $C_t \in \C'$:
	\begin{enumerate}[(i)]
		\item \label{it:otc:trans:opop} there must be a pair of determining vertices for either $S$ or $\overline{S}$; or %
		\item \label{it:otc:trans:opvert} there must be a pair of opposing vertices for $S$ and $\overline{S}$. %
	\end{enumerate}
	
\end{lemma}

\begin{proof}
	Since $(S,\overline{S})$ is a valid solution, then for each translated curve $C_t \in \C'$, there must exist a partition of $C_t$ defined by an index $i$ such that $C_t[1,i] \subset S$ and $C_t[i+1,m] \subset \overline{S}$.
	Assume that neither of the conditions stated in the lemma hold.
	Then the radius of the squares can be decreased to obtain a smaller pair of squares coincident with the same corners of $R$.
	If no vertices from the curves in $\C'$ lie on the inner sides of $S$ and $\overline{S}$---that is, the sides that are not collinear with sides of $R$---%
	then the radius can be reduced without translating the curves in $\C'$.
	If one or more prefix (suffix) vertices of lie on the inner sides of $S$ ($\overline{S}$), then $C_t$ is translated in a direction determined in the following way.
	For each such vertex $p$ lying on a side $s$ of its assigned square, let $\vec{n}$ be the direction of the inner normal of $s$.
	The direction of translation is the direction of the vector obtained by summing the normal vectors.
	Such a direction would allow all the vertices lying on the sides of their respective squares to remain on the side, unless two vertices lie on opposing sides of the same square, i.e., condition~(\ref{it:otc:trans:opop}) holds, or they lie on the opposing inner sides of different squares, i.e., condition~(\ref{it:otc:trans:opvert}) holds.
\end{proof}

\Cref{lem:otc:trans:optvert} implies that the optimality of a solution will be determined by the partition of a single curve. %
The minimum radius of a solution for a partition at $i$ of a curve $C_j$ under translation may be computed in constant time by finding the bounding boxes around the prefix and suffix of the curve, and the radius of the solution can then be obtained from the candidate pairs of determining and opposing vertices implied by the bounding boxes.
Specifically, the value $r_i^j$ is a lower bound on the optimal radius obtained by the partition at $i$ of curve $C_j$, and can be computed in constant time, for example, when $S$ is \emph{below-left} of $\overline{S}$: %
\[r_i^j := \frac{1}{2} \min \left\{
\begin{array}{lr}
\width(C_j[1,i]), 
\width(C_j[i+1,m]),
(\delta_x^* - (\min_{v \in C[i+1,m]} v.x - \max_{v \in C[1,i]} v.x))/2, \\
\height(C_j[1,i]), 
\height(C_j[i+1,m]), 
(\delta_y^* - (\min_{v \in C[i+1,m]} v.y - \max_{v \in C[1,i]} v.y))/2
\end{array}
\right\}
.\]

An optimal solution for $\C$ under translation where the squares coincide with a particular pair of opposing corners of $R$ can computed as $r := \max_{j \colon C_j \in \C} \min_{1\leq i \leq m} r_i^j$, i.e., the minimum radius of a pair of squares covering the partition of a curve, and then determining the largest such value over all curves.
The solutions are evaluated where $S$ and $\overline{S}$ coincide with  each of the four ordered pairs of opposite corners of $R$, and the overall solution is the smallest of these values.
We thus obtain the following result.

\begin{restatable}{theorem}{thmotctranslation}
    Given a set of curves $\C$ as input, an optimal solution to the \otc\ problem under translation using the discrete \frechet\ distance under the $L_\infty$ metric can be computed in $O(nm)$ time and $O(nm)$ space.
\end{restatable}

\subsection{\otc{} and \texorpdfstring{$L_{2}$}{L-2} metric}
\label{sec:otc:l2}

For the \otc{} problem and $L_2$ we need some more sophisticated arguments, but again we use a similar basic approach.

We first consider the decision problem: Given a value $r > 0$, determine whether there exists a segment $s$ such that $\max_{C_i \in \C} \dfd(s,C_i) \le r$.

For each curve $C \in \C$ and for each vertex $p$ of $C$, draw a disk of radius $r$ centered at~$p$ and denote it by $D(p)$. 
Let $\D$ denote the resulting set of $nm$ disks and let $\A(\D)$ be the arrangement of the disks in $\D$. The combinatorial complexity of $\A(\D)$ is $O(n^2m^2)$. 
Let $A$ be a cell of $\A(\D)$. Then, each curve $C=(p_1,\ldots,p_m) \in \C$ induces a bit vector $V_C$ %
of length~$m$; the $i$th bit of $V_C$ is 1 if and only if $D(p_i) \supseteq A$. Moreover, if $j$ is the index of the first 0 in~$V_C$, then {\em the suffix of curve $C$ at cell $A$} is $C[j,m]$.   

We maintain the vectors $V_C$ as we traverse the arrangement $\A(\D)$, by constructing a binary tree $T_C$, for each curve $C$, as described in the previous section. The leaves of $T_C$ correspond to the vertices of $C$, and in each node we store a single bit. Here, the bit at a leaf node corresponding to vertex $p_i$ is 1 if and only if $D(p_i) \supseteq A$, where $A$ is the current cell of the arrangement. For an internal node, the bit is 1 if and only if all the bits in the leaves of its subtree are 1. We can determine the current suffix of $C$ in $O(\log m)$ time, and the cost of an update operation is $O(\log m)$.
We also maintain the set $P$, where $P$ is the union of the suffixes of the curves in $\C$, and its corresponding region $X = \cap_{p \in P} D(p)$. Actually, we only need to know whether $X$ is empty or not.

We begin by constructing the trees $T_{C_1},\ldots,T_{C_n}$ and initializing all bits to 0, which takes $O(mn)$ time. We also construct the data structures for $P$ and $X$, where initially $P=C_1[1,m] \cup \cdots \cup C_n[1,m]$. This takes $O(nm \log^2 (nm))$ time in total. For $P$ we use a standard balanced search tree, and for $X$ we use, e.g., the data structure of Sharir~\cite{Sharir97}, which supports updates to $X$ in $O(\log^2 (nm))$ time. 
We now traverse $\A(\D)$ systematically, %
beginning with the unbounded cell of $\A(\D)$, which is not contained in any of the disks of $\D$.
Whenever we enter a new cell $A$ from a neighboring cell separated from it by an arrangement edge, then we either enter or exit the unique disk of $\D$ whose boundary contains this edge.
We thus first update the corresponding tree $T_C$ accordingly, and redetermine the suffix of $C$. We now may need to perform $O(m)$ update operations on the data structures for $P$ and $X$, so that they correspond to the current cell.
At this point, if $X \neq \emptyset$, then we halt and return \textsc{yes} (since we know that the minimum enclosing disk of the union of the prefixes is at most~$r$). If, however, $X = \emptyset$, then we continue to the next cell of $\A(\D)$, unless there is no such cell in which case we return \textsc{no}. 
We conclude that the decision problem can be solved in $O(n^2m^3\log^2 (nm))$ time and $O(n^2m^2)$ space.

Notice that the minimum radius $r^*$ for which the decision version returns \textsc{yes}, is determined by three of the $nm$ curve vertices. Thus, we perform a binary search in the (implicit) set of potential radii (whose size is $O(n^3m^3)$) in order to find $r^*$. Each comparison in this search is resolved by solving the decision problem for the appropriate potential radius. Moreover, after resolving the current comparison, the potential radius for the next comparison can be found in $O(n^2m^2 \log^2 (nm))$ time, as in the early near-quadratic algorithms for the well-known 2-center problem, see, e.g.,~\cite{AgarwalS94,Jaromczykk94,KatzS97}.

The following theorem summarizes the main result of this section.

\begin{restatable}{theorem}{thmotcltwo}
	Given a set of curves $\C$ as input, an optimal solution to the \otc\ problem using the discrete \frechet\ distance under the $L_2$ metric can be computed in $O(n^2m^3\log^3 (nm))$ time and $O(n^2m^2)$ space.
\end{restatable}

\bibliographystyle{alphaurlinit}
\bibliography{refs}

\end{document}